\documentclass[a4paper,UKenglish,cleveref, autoref, thm-restate]{article}
\usepackage[a4paper, total={6in, 10in}]{geometry}
\usepackage{amsfonts}

\pdfoutput=1

\newcommand{\cryptograph}{Crypto'Graph}
\title{\cryptograph: Leveraging Privacy-Preserving Distributed Link Prediction for Robust Graph Learning } 

\author{Sofiane Azogagh, \href{mailto:birba.zelma_aubin@courrier.uqam.ca}{Zelma Aubin Birba}, Sébastien Gambs and Marc-Olivier Killijian}

\date{}

\newcommand{\V}{\mathcal{V}}
\newcommand{\E}{\mathcal{E}}
\newcommand{\X}{\mathcal{X}}
\newcommand{\Y}{\mathcal{Y}}
\newcommand{\A}{\mathcal{A}}
\newcommand{\Z}{\mathbb{Z}}
\newcommand{\G}{\mathbb{G}}
\renewcommand{\P}{\mathbb{P}}
\newcommand{\CN}{\operatorname{CN}}
\newcommand{\J}{\operatorname{J}}
\newcommand{\Cosine}{\operatorname{Cosine}}
\newcommand{\relu}{\operatorname{ReLu}}
\newcommand{\softmax}{\operatorname{softmax}}

\usepackage{comment}
\usepackage{hyperref}
\usepackage{multirow}

\usepackage{stmaryrd}
\usepackage{xcolor}
\usepackage{tikz}
\usetikzlibrary{matrix,shapes,arrows,positioning,chains, calc}
\usepackage[most]{tcolorbox}
\newtcolorbox{alachemfig}[2][]{
    enhanced,
    size=fbox,sharp corners,
    colback=white,colframe=black,
    colbacktitle=black,fonttitle=\bfseries,
    attach boxed title to top center={yshift=-3mm,yshifttext=-3mm},
    boxed title style={size=small,left=0pt,right=0pt,sharp corners},
    title=#2,
    width=#1
    }

\begin{document}
\newtheorem{theorem}{Theorem}
\newtheorem{definition}{Definition}
\newtheorem{proof}{Proof}

\maketitle

\def\thefootnote{}
\footnotetext{This work is supported by the DEEL Project CRDPJ 537462-18 funded by the National Science and Engineering Research Council of Canada (NSERC) and the Consortium for Research and Innovation in Aerospace in Québec (CRIAQ), together with its industrial partners Thales Canada inc, Bell Textron Canada Limited, CAE inc and Bombardier inc. \url{https://deel.quebec}}

\def\thefootnote{\arabic{footnote}}

\begin{abstract}
Graphs are a widely used data structure for collecting and analyzing relational data. 
However, when the graph structure is distributed across several parties, its analysis is particularly challenging. 
In particular, due to the sensitivity of the data each party might want to keep their partial knowledge of the graph private, while still willing to collaborate with the other parties for tasks of mutual benefit, such as data curation or the removal of poisoned data.
To address this challenge, we propose \cryptograph{}, an efficient protocol for privacy-preserving link prediction on distributed graphs. 
More precisely, it allows parties partially sharing a graph with distributed links to infer the likelihood of formation of new links in the future. 
Through the use of cryptographic primitives, \cryptograph{} is able to compute the likelihood of these new links on the joint network without revealing the structure of the private individual graph of each party, even though they know the number of nodes they have, since they share the same graph but not the same links.
\cryptograph{} improves on previous works by enabling the computation of a certain number of similarity metrics without any additional cost.
The use of \cryptograph{} is illustrated for defense against graph poisoning attacks, in which it is possible to identify potential adversarial links without compromising the privacy of the graphs of individual parties. 
The effectiveness of \cryptograph{} in mitigating graph poisoning attacks and achieving high prediction accuracy on a graph neural network node classification task is demonstrated through extensive experimentation on a real-world dataset.
\end{abstract}

\section{Introduction}
\label{sec:intro}

In today's digital age, graphs have emerged as the predominant format for representing relational data, as they naturally capture both the relationships and structures inherent in such datasets. 
Indeed, from social networks~\cite{Wilson09} to biological systems~\cite{pavlopoulos2011using}, the interconnection of entities can be easily visualized and understood through graphs.
However, as data becomes increasingly distributed, a new set of challenges arises with respect to their analysis. 
For example, in a scenario where a graph's structure is distributed across multiple parties, the goal might be to study this structure without any party disclosing the private details of their segment. Such an analysis could involve predicting potential future links \cite{Liben03, Zheng15, Zhang18, Demirag23} or identifying malicious links that an adversary has introduced to compromise the graph's integrity \cite{Zugner18, Wu19, Xu23}. As such attacks might happen without been noticed, it is crucial to act preventively and allow for collaboration to defend against them.

To address these issues, we propose \cryptograph{}, a novel protocol designed for privacy-preserving link prediction on distributed graphs. 
To avoid privacy leakage, \cryptograph{} leverages cryptographic primitives such as Diffie-Hellman shared secrets and Private Set Intersection Cardinality (PSI-CA), which ensure that likelihood similarities used for link prediction can be computed on the joint network without exposing the specifics of the private individual graphs.
Furthermore, \cryptograph{} can be used as a robust defense against graph poisoning attacks.
More precisely, by predicting potential links without jeopardizing the confidential information of individual nodes, it can be used to effectively detect adversarial links, improving the quality of downstream graph learning tasks. 

The main contributions of this paper are:
\begin{itemize}
    \item We propose \cryptograph{}, a new protocol for distributed privacy-preserving link prediction on graph data via the computation of the common neighbors heuristic.
    \cryptograph{} is more efficient, by several orders of magnitude, than state-of-the-art methods while making it possible to derive other metrics such as the Jaccard and Cosine similarity measures at no extra cost, thus allowing to choose among different heuristics to adapt the link prediction to the actual data. 
    Finally, unlike previous works, it can be used on the complete graph without compromising the privacy of the private individual graphs as we demonstrate by proving its security against  
    graph reconstruction attacks.
    \item The applicability of \cryptograph{} is illustrated through a collaborative defense against graph poisoning scenario. 
    More precisely, we show how to leverage our protocol to privately derive link likelihood information in a distributed manner, enabling the different parties to identify adversarial links and remove them for a better utility on subsequent tasks, namely the training of graph neural networks. 
    In addition, we show that the benefit of collaborating via our protocol varies according to the common knowledge between the participants, the amount of adversarial links introduced as well as the type of attack conducted.
    Nonetheless, experiments on a real dataset demonstrate that it is almost always beneficial to cooperate even when the data of one party has not been poisoned. This encourages the use of our solution without prior certainty that one or both of the private graphs have undergone an attack.
\end{itemize}

The outline of the paper is as follows.
First, in Section~\ref{sec:soa}, we review the related work on link prediction both in the centralized and distributed settings before introducing in Section~\ref{sec:prelim} the background notions on link prediction, graph neural network and private set intersection, that are necessary to the understanding of our work. 
Afterwards, in Section~\ref{protocol} we describe \cryptograph{}, our protocol for secure and distributed link prediction before detailing how \cryptograph{} can be used to defend against graph poisoning in Section~\ref{sec:application-gnn}, in which we evaluate it on a real-world graph dataset.

\section{Related work}
\label{sec:soa}

Based on the structure of the graph and potential additional information (such as the values of node attributes), link prediction algorithms \cite{Liben03, Zheng15, Zhang18, Demirag23} aim at identifying future probable links in dynamic networks.
One of the first proposed methods for predicting a link between two nodes consists in measuring the similarity of these nodes by leveraging their neighborhood structure and, based on the assumption that nodes that have common neighbors tend to create connections, predict new links or not. 
This similarity can be computed on the structural information between nodes but also by considering their attributes. 
For instance, in research collaboration networks, Newman has shown that, in domains such as physics, the more coauthors two researchers have in common (\emph{i.e.}, the more common neighbors they have), the more they are likely to collaborate in the future~\cite{Newman01}. 
In addition, he has also observed that a scientist taken at random is more likely to make new collaborations if he has many past ones, introducing the notion of preferential attachment. 
Other similarity measures such as the Jaccard similarity~\cite{Jaccard01}, the Cosine similarity~\cite{Singhal01} or the Adamic-Adar index~\cite{Adamic03} have been explored as well.

Machine learning-based models have also been proposed to tackle the link prediction problem. 
For instance, Kashima and Abe have extracted topological features from the graph and used them to train a model for supervised link prediction~\cite{Kashima06}. 
Zhang and Chen have designed \textit{SEAL}, a neural network-based architecture for capturing the link formation law on a graph~\cite{Zhang18}. 
In particular, they have proven that low order subgraphs, composed only of few hops-neighbors, are often sufficient to estimate high order metrics that reason on the whole graph. 
They have also proposed a framework for predicting new links using network-based heuristics applied on these subgraphs.
Other link prediction methods also try to learn informative features for nodes while avoiding to explore the whole graph. 
For example, Perozzi and collaborators have used random walks in conjunction with the \textit{SkipGram} model~\cite{Mikolov13} to build representations of nodes based on samples of their neighborhoods~\cite{Perozzi14}. 
The Node2vec embedding technique~\cite{Grover16} also relies on an analog flexible approach for neighborhood sampling in the graph to generate continuous node feature representations that can then be compared between nodes to predict new links. 

However, the aforementioned methods are focused on the centralised setting, in which there is a single graph in the hands of only one party. 
Nonetheless, subsequent works have proposed approaches inspired by the previous ones but adapted to the multi-graph or distributed graph setting. We define multi-graph link prediction as the setting where the parties own graphs with potentially different nodes and links. The distributed graph setting, on the other hand, assumes that the parties hold the same nodes in their graphs, but not necessarily the same links.
For instance, some works have proposed to allow members of a decentralised social network like Mastodon \footnote{\href{https://joinmastodon.org}{https://joinmastodon.org}}, to define privacy controls on their connections by specifying which of their friendships they are willing to disclose~\cite{Zheng15}. 
The service provider can then use this information to train a logistic regression model that privately makes friendship recommendations. 
Another recent work has considered the distributed graph setting and allows multiple subgraph owners to make link predictions by computing similarity metrics in a secure manner~\cite{Zhang23}. 
More precisely by using secret sharing techniques, it is possible to privately aggregate the local similarity scores and allow the parties to make their decision based on the private aggregate. 
However, their approach can induce an accuracy loss in the prediction because it does not take into account the cross-party similarities (similarity between a node $x$ in one subgraph, and a node $y$ in  another subgraph). 
Another recent approach~\cite{Demirag23} computes the common neighbors similarity measure using three instances of a Private Set Intersection (PSI) protocol , which we will detail later in Section~\ref{preliminaries:psi}. 
However, it is not clear how this protocol can be adapted for the computation of other similarity measures and its computational cost is higher than the approach we proposed.

In the remainder of this paper, we introduce a method that addresses the drawbacks of the two former works. 
More precisely, our protocol improves on the accuracy compared to~\cite{Zhang23} while allowing to compute a wide range of similarity measures and thus not being limited to common neighbors as in~\cite{Demirag23}. 
In addition, as our protocol is highly efficient this makes it suitable for large-scale link prediction on a large graph, which motivates its usage as a defense mechanism against graph poisoning, as described later in Section~\ref{sec:application-gnn}.

\section{Preliminaries}
\label{sec:prelim}

In this section, we start by providing the notations we will be using, then we introduce the preliminary notions of link prediction, graph neural network and private set intersection, that are necessary to the understanding of the remainder of this paper.

To begin, the notations employed throughout the paper are presented in Table~\ref{table-notation}.

\begin{table}[h!]
\begin{tabular}{|l|l|l|}
\hline
\multirow{3}{*}{Graph}  & $\V$ & set of vertices (\emph{i.e} nodes) \\ \cline{2-3} 
                        & $\E$ & set of edges (\emph{i.e} links)  \\ \cline{2-3} 
                        & $G=(\V,\E)$ & graph of nodes $\V$ and links $\E$   \\ \cline{2-3}
                        & $\Gamma(x) \subseteq\V$ & the neighbors of $x$ in $G$ (\emph{i.e} the nodes in $\V$ that share a link with $x$) \\ \hline
\multirow{2}{*}{GNN}    & $\X \in \mathbb{R}^{d\times |\V|}$      & feature matrix  where $\X[i]$ is $d-$dimensional feature vector of node $i$ \\ \cline{2-3} 
                        & $\Y \in \mathbb{R}^{|\V|}$       & labels of nodes where $\Y[i]$ is the label of node $i$ \\ \cline{2-3}
                        & $G = (\V,\E,\X,\Y$) & extended graph (with features and labels)  \\ \cline{2-3}
                        & $\A$ & adjacency matrix of a graph \\ \cline{2-3}
                        & $\tilde{A}$ & adjacency matrix with self-connections inserted (\emph{i.e} $\tilde{A} = A + I_{\V})$ \\ \cline{2-3}
                        & $\tilde{D}$ & degree matrix (\emph{i.e} $\tilde{D_{ii}} = \sum_{j} \tilde{A}_{ij}$) at a specific layer \\ \hline
\multirow{4}{*}{Crypto} & $p,q$       & large prime $p$ and integer $q$ such that $q$ divides $p-1$\\ \cline{2-3} 
                        & $\mathbb{Z}_p$      & set of integers modulo p \\ \cline{2-3} 
                        & $\G_q$      & multiplicative group of order $q$ \\ \cline{2-3} 
                        & $g$       & generator of $\G_q$ \\ \hline

\end{tabular}

\caption{A guide to the notation used in this paper.}
\label{table-notation}
\end{table}

\subsection{Link prediction}
\label{sec:linkprediction}

Link prediction~\cite{Newman01,Liben03} is a graph learning task that aims to infer the potential existence of links between nodes that are not currently connected. 
Link prediction has many possible applications, such as friend recommendation in a social network~\cite{Liben03} or in healthcare for the the study of contacts for epidemic control purposes~\cite{antweiler2022uncovering}.
A lot of link prediction methods are based on the computation of the similarity between the neighborhoods of pairs of nodes. 
More formally, considering two nodes $x,y \in \V$, we can compute their similarity in the following ways (which are also the most commonly used in the literature):
\begin{itemize}
    \item \textbf{Common Neighbors:} The common neighbors similarity of $x$ and $y$ is defined as $$\CN(x,y) = | \Gamma(x) \cap \Gamma(y)|$$.
    \item \textbf{Jaccard:} The Jaccard similarity between $x$ and $y$ is defined as $$\J(x,y) =\frac{|\Gamma(x) \cap \Gamma(y)|}{|\Gamma(x) \cup \Gamma(y)|}= \frac{\CN(x,y)}{|\Gamma(x) \cup \Gamma(y)|}$$.
    \item \textbf{Cosine:} The cosine similarity between $x$ and $y$ is given by $$\Cosine(x,y) = \frac{|\Gamma(x) \cap \Gamma(y)|}{\sqrt{ |\Gamma(x)| } \times \sqrt{|\Gamma(y)|}} = \frac{\CN(x,y)}{\sqrt{ |\Gamma(x)| } \times \sqrt{|\Gamma(y)|}}$$
\end{itemize}

In real-world scenarios, the graph might be distributed among different parties. 
To account for this, and without loss of generality, we consider that the whole graph $G=(\V,\E)$ is distributed among two parties $P_1$ and $P_2$ such as each party entirely knows $\V$ but only a fraction of $\E$. It's worth noting that this distribution can be generalized, as shown in~\cite{Zhang23}.
Let $G_1(\V,\E_1)$ denote the graph of $P_1$ and $G_2(\V,\E_2)$ the graph of $P_2$ such as $\E_1 \subset \E$ and $\E_2 \subset \E$.
We consider the scenario in which $P_1$ and $P_2$ want to collaborate to predict the existence of a link in their respective graphs without revealing them due to confidentiality issues that may arise (\emph{e.g.}, the graphs may represent personal relationships). 
A possible solution to predict a link between $x,y \in \V$ is for $P_1$ to privately share $\Gamma_1(x) $ and $ \Gamma_1(y)$ with $P_2$, who in turn privately shares $\Gamma_2(x)$ and $ \Gamma_2(y)$. 

Subsequently, the link prediction primitive computes the similarity measure on the joined graph $G$ and outputs this measure to both $P_1$ and $P_2$. 
Finally, each party decides to predict a link based on this result and a threshold chosen independently as illustrated in Figure \ref{fig:linkprediction}.

\begin{figure}[h!]
    \centering
    \includegraphics[width=\linewidth]{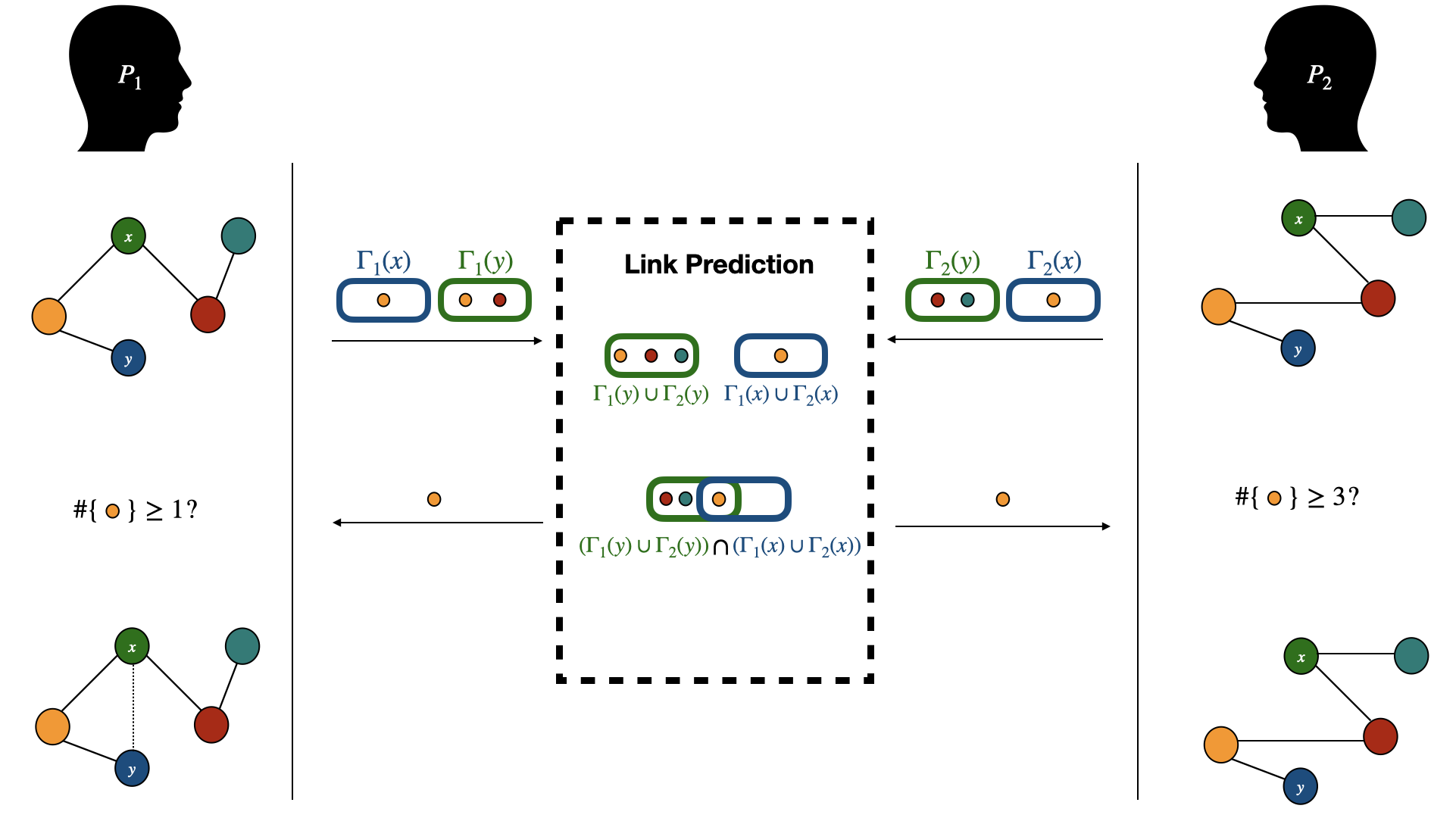}
    \caption{An illustration of a collaboration between $P_1$ and $P_2$ to predict a link between the node $x$ (green) and $y$ (blue). 
    Here, the link prediction primitive will output $\CN(x,y)$ that each party $P_i$ will compare with their personal threshold, which is $1$ for $P_1$ and $3$ for $P_2$. 
    Based on the result of this comparison, a link between $x$ and $y$ can be added or not, in one or both graphs $G_1$ and $G_2$.}
    \label{fig:linkprediction}
\end{figure}

\subsection{Graph Neural Networks}

To take advantage of their structural information, deep learning approaches have been adapted to the graph data format with the advent of Graph Neural Networks (GNNs)~\cite{Scarselli09, Kipf17}, which have made it possible to harness the structure and the node properties in graphs for various learning tasks. 
In a nutshell, most GNN models take as input a graph  $G=(\V,\E,\X, \Y)$ and combine the features of each node with the ones of its neighbors to perform the predictions. 
In this way, new node features are created with each node aggregating information from farther in the graph at each layer.
Nonetheless, different types of models have been introduced over the years.  
For instance, representation learning algorithms~\cite{Perozzi14, Grover16} are deep learning-based architectures that represent nodes, edges or subgraphs of a graph as a multidimensional vector. 
This vector can then be used for various statistical and prediction tasks, such as community identification (\emph{i.e.}, a form of clustering) and link prediction. 

In this work, we focus on Graph Convolutional Networks (GCNs)~\cite{Kipf17}, a specific type of network that mimics convolutions on images. 
More precisely, we consider a semi-supervised node classification task that aims at predicting the labels of test nodes based on the features, edges and labels of a set of training nodes. 
Following the approach taken in~\cite{Kipf17, Zugner18}, we consider a two layer neural network given by the formula : 
 $$ Z = \softmax\left(\hat{A} \ \relu \left( \hat{A}XW^{(0)}\right) W^{(1)}\right),  $$
in which $\hat{A} = \tilde{D}^{-\frac{1}{2}} \tilde{A} \tilde{D}^{-\frac{1}{2}}$, and $W^{(0)}, W^{(1)}$ are the trainable parameters of the network and $Z$ is the prediction of the network for a given node.

Similarly to classical deep learning models, GNNs have been shown to be vulnerable to adversarial attacks. 
More specifically for the task considered here, this means that an adversary can poison the graph by injecting links or altering node features to influence the prediction made by the model~\cite{Zugner18, Wu19}. 
Some of the classical defenses against such poisoning attacks leverage link prediction techniques~\cite{Wu19, Xu23}, in which by computing a likelihood score for each edge, one can identify the unlikely links and consider them as being potentially adversarial (and thus cleaning them).

\subsection{Private Set Intersection}
\label{preliminaries:psi}

A Private Set Intersection (PSI) protocol~\cite{Freedman04} allows two or more parties, each holding a private set of items to compute the intersection of their sets (\emph{i.e.}, items that they have in common) without revealing any additional information.
PSI has found many applications in the real world, such as private data mining~\cite{Nomura20}, the analysis of genomics or medical data~\cite{Aziz17,miyaji2017privacy,Shen18,Ruan19} or even botnet detection~\cite{Nagaraja10}. 
Since it was first introduced, PSI has evolved in many subvariants of protocols according to the scenario in which it is applied.
For example, if only one of the parties needs to obtain the intersection at the end of the computation, it is called a \textit{one-way} PSI, while otherwise it is a \textit{mutual} PSI. 
Another type of PSI protocol was developed according to the size the sets of each party.
For instance, if the sizes of the private sets are similar, it is refer to as a \textit{balanced} PSI while otherwise it is called an \textit{unbalanced} one. 
More generally, a classification has recently been proposed in~\cite{MORALES2023}. 

In other scenarios, the parties may want to know \textit{how much} they share but not what exactly they have in common. 
This can be solved using a \textit{PSI-CA} (standing for Private Set Intersection CArdinality), which cannot be directly instantiated from classical PSI. 
Among the cryptographic building blocks that can be combined to develop a PSI-CA protocol, we can cite for instance homomorphic encryption, as seen in works like~\cite{Debnath17, Liu18, Ion21}, oblivious transfer \cite{Dong17}, using generic public key techniques described in~\cite{Debnath20, ResendeA18}, or even commutative encryption similar to the method outlined in \cite{Siyi20}. 
In this latter approach, one party, $P_1$, initiates the protocol by sending its own set of items $X$ encrypted as $Enc_{PK_1}(X)$ to $P_2$. 
Afterwards $P_2$ shuffles the elements and send them back to $P_1$ as $Enc_{PK_2}(Enc_{PK_1})(X)$. 
Additionally, $P_2$ sends its own set encrypted as $Enc_{PK_2}(Y)$. 
By using a commutative encryption scheme, $P_1$ can ``delete'' its corresponding key $PK_1$ from $Enc_{PK_2}(Enc_{PK_1}(X))$ to obtain $Enc_{PK_2}(X)$ and then compare the number of correspondence with $Enc_{PK_2}(Y)$. This approach inspired a lot of PSI-CA including the one that we use in this paper which is based on~\cite{decristofaro12}.

\section{Protocol}
\label{protocol}
In this section, we present our protocol \cryptograph{}, which enables to perform link prediction on graph data between two parties in a distributed and privacy-preserving manner. 
However, our protocol can easily be generalized to more than two parties by following the approach proposed in~\cite{Zhang23}.

\subsection{Description}
\cryptograph{} is close in spirit to the PSI functionality, in the sense that it privately computes the common neighbors of two nodes on a distributed graph. 
Let $G_1=(\V,\E_1)$ and $G_2=(\V,\E_2)$ be the private subgraphs owned by the semi-honest parties $P_1$ and $P_2$. 
The main idea behind our approach is the following : by representing the union graph $G_1 \cup G_2$ as a data structure that masks the neighborhood of each node while keeping its cardinality, it is possible to count the common neighbors of any two nodes. 
To achieve this objective, we propose a solution based on the use of Diffie-Hellman shared secrets~\cite{Diffie76}. Our solution can be divide in two phases: an offline phase that can be precomputed before the beginning of the concrete protocol, and the real execution of the protocol on the precomputed data.

More precisely, for a link prediction performed between nodes $x$ and $y$, each party represents the neighborhoods of the nodes as sets $\{g^{\alpha \beta x_i}\}_{i \in \{1,\ldots,|\Gamma_k(x)|\} }$ and $\{g^{\alpha \beta y_i}\}_{i \in \{1,\ldots,|\Gamma_k(y)|\} }$, with $k \in \{1, 2\}$ and in which $\alpha$ and $\beta$ are secrets randomly sampled respectively by $P_1$ and $P_2$. 
Those secrets have the property that they can be used to compare the sets without disclosing the individual elements, which we leverage to compute the oblivious union of the neighbors of $x$ in both graphs, as well as for node $y$. 
Afterwards, we count the common elements of those sets to determine the size of the intersection.

One of the strength of our method is that by computing the intermediary sets of neighbors of $x$ and $y$ separately, those results can also be reused to compute the Jaccard similarity by dividing the common neighbors score with the size of the union of the neighbors of $x$ and $y$.
More generally, our method allows for the computation of any similarity metric involving the sizes of the immediate neighborhood of two nodes~\cite{Newman01, Jaccard01}. 

As an additional contribution, we propose an optimization of the aforementioned algorithm based on a caching mechanism.
More precisely, by keeping the same $\alpha$ and $\beta$ keys across predictions, for each node $t$ that has been involved in a previous prediction, we can reuse its encryption $g^{\alpha \beta t}$. 
A detailed analysis of this caching mechanism and its security are provided in Section~\ref{sec:security}.
Figure~\ref{fig:privatelinkprediction} provides a graphic illustration of the \cryptograph{} protocol, where the offline phase is represented in gray, and the online one in black.

\begin{figure}[h!]
\begin{alachemfig}[\textwidth]{\cryptograph{} Private Link Prediction}

\begin{tikzpicture}
\matrix (m)[matrix of nodes, column  sep=0.5cm,row  sep=6mm, nodes={draw=none, anchor=center,text depth=0pt} ]{
$P_1$ : $G_1=(\V,\E_1)$                           & & $P_2$ : $G_2=(\V,\E_2)$\\[-5mm]
$\color{gray} \alpha \gets \Z_q$                                 & & $\color{gray}\beta \gets \Z_q$  \\[-7mm]
$\color{gray} \forall x_i \in \Gamma_1(x), a_i = g^{\alpha x_i}$ & & $\color{gray}\forall x_i \in \Gamma_2(x), c_i = g^{\beta x_i}$\\[-7mm]
$\color{gray} \forall y_i \in \Gamma_1(y), b_i = g^{\alpha y_i}$ & & $\color{gray}\forall y_i \in \Gamma_2(y), d_i = g^{\beta y_i}$\\[-7mm]
                                    & $\{a_1,\dots,a_{|\Gamma_1(x)|}\}$ & \\[-5mm]
                                    & $\{b_1,\dots,b_{|\Gamma_1(y)|}\}$ &
                                                    & & \\[-10mm]
                                                    & & $\forall a_i \in \{a_1,\dots,a_{|\Gamma_1(x)|}\}, a_i'=a_i^\beta$  \\[-7mm]
                                                    & & $\forall b_i \in \{b_1,\dots,b_{|\Gamma_1(y)|}\}, b_i'=b_i^\beta$  \\[-7mm]  
                                    & $\{a_1', \dots, a_{|\Gamma_1(x)|}' \}$ &\\[-7mm]
                                    & $\{b_1', \dots, b_{|\Gamma_1(y)|}' \}$ &\\[-5mm]
                                    & $\{c_1, \dots, c_{|\Gamma_2(x)|} \}$ &\\[-7mm]
                                    & $\{d_1, \dots, d_{|\Gamma_2(y)|} \}$ &\\[-5mm]
                                                    & & \\[-5mm]
$\forall c_i \in \{c_1,\dots,c_{|\Gamma_2(x)|}\}, c_i'=c_i^\alpha$ & & \\[-7mm]
$\forall d_i \in \{d_1,\dots,d_{|\Gamma_2(y)|}\}, d_i'=d_i^\alpha$ & & \\[-4mm]
$\llbracket \Gamma(x) \rrbracket = \{a_1',\dots\} \cup \{c_1',\dots\}$     & & \\[-7mm]
$\llbracket \Gamma(y) \rrbracket = \{b_1',\dots\} \cup \{d_1',\dots\}$     & & \\[-4mm]
$\CN(x,y) = |\llbracket \Gamma(x) \rrbracket \cap \llbracket \Gamma(y) \rrbracket|$           & & \\[-7mm]
                                                &$\CN(x,y)$ & \\
                                                  & & \\
};
\draw[shorten <=-1cm,shorten >=-0.5cm] (m-1-1.south east)--(m-1-1.south west);
\draw[shorten <=-0.5cm,shorten >=-1cm] (m-1-3.south east)--(m-1-3.south west);
\draw[shorten <=-0.5cm,shorten >=-0.5cm,-latex] (m-6-2.north west)--(m-6-2.north east);
\draw[shorten <=-0.5cm,shorten >=-0.5cm,-latex] (m-11-2.north east)--(m-11-2.north west);
\draw[shorten <=-0.5cm,shorten >=-0.5cm,-latex] (m-19-2.south west)--(m-19-2.south east);
\end{tikzpicture}
\end{alachemfig}

\caption{\small{A diagram of \cryptograph{}, our private link prediction protocol for two nodes. 
Both parties are assumed to share a common element $g \in \mathbb{G}_q$. 
In the offline part (gray), both parties generates a key ($\alpha$ for $P_1$ and $\beta$ for $P_2$) and then encrypt the neighbors of the nodes of $x$ and $y$ of their respective graphs. 
The online phase (black) begins with $P_1$ transmitting the encrypted nodes under consideration to $P_2$. 
Subsequently, $P_2$ proceeds by re-encrypting them using his own key $\beta$ before shuffling them at random and sending them accompanied of his own encrypted nodes. 
Hence, leveraging the commutativity of the encryption algorithm, $P_1$ can incorporate his key to the nodes of $P_2$. Finally once $P_1$ gets $|\Gamma(x)|$ and $|\Gamma(y)|$ by matching the ciphertexts, he can compute the similarity measure (here Common Neighbor, but it can also be Jaccard and Cosine as well). }}
\label{fig:privatelinkprediction}
\end{figure}
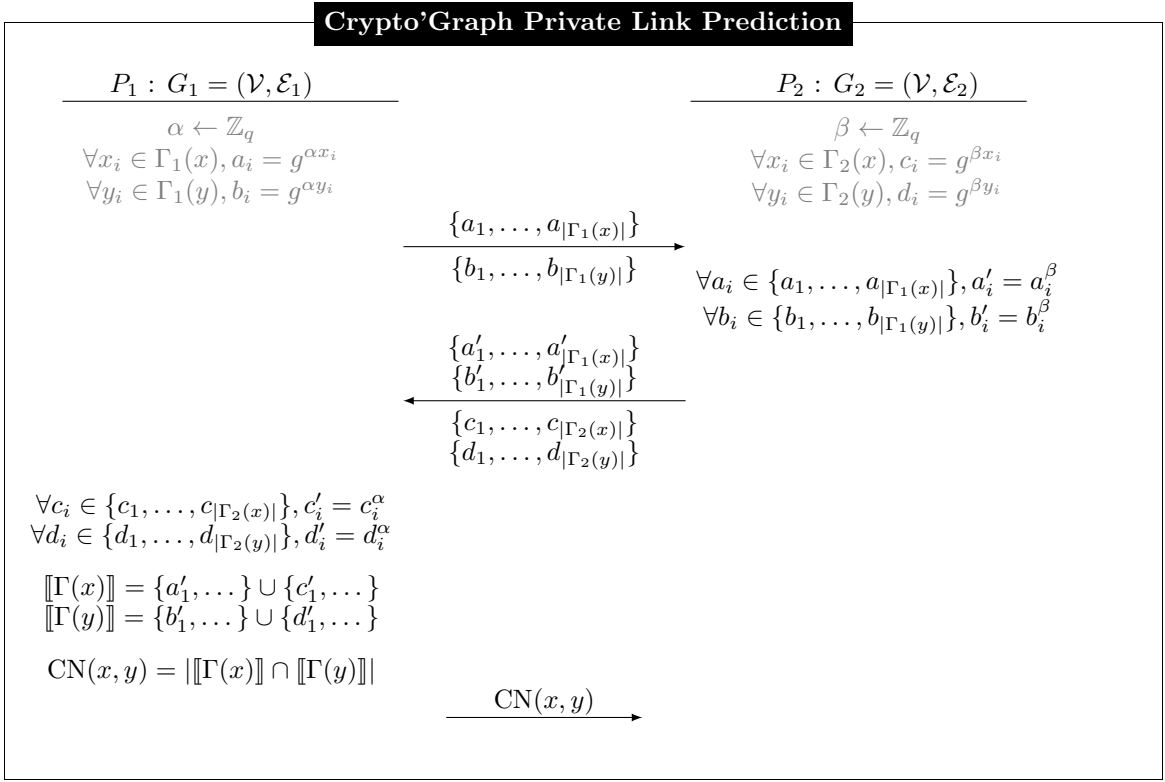

\subsection{Performance results}
\label{protocol:performance}
To assess the performance of our protocol, we evaluate it on a real world graph commonly used in the literature for link prediction. 
Our experiments are conducted on two subgraphs of the Polblogs dataset from~\cite{nr}, which represents political blogging sites in the USA. 
The two subgraphs are obtained by sampling links from the original dataset. 
More precisely, as in \cite{Zhang23}, the membership of a link to one or both of the graphs is determined as follows : we choose two values $q_1, q_2$ in the interval $[0,1]$, such that $q_1 \leq q_2$.
Then, for each link in the original graph, we sample a random value $v$ in the same interval. 
If the sampled value lands in $[0,q_1]$, the link is attributed to $G_1$ while if $v \in ]q_1, q_2]$, the link is added to $G_2$. Otherwise the link is attributed to both the subgraphs. 

This generation of the distributed dataset enables to control the proportions of links owned by one or both of the graphs. 
The experiments described hereafter are made with $q_1$ and $q_2$ set respectively to 0.3 and 0.6, which results in each subgraph solely owning 30\% and sharing 40\% of the initial graph with the other one. 

Our implementation is single-threaded and developed in C++, and for efficient exponentiation, we use the OpenSSL implementation of the \textit{NIST P-256}~\cite{sec2} elliptic curve. 
This choice allows to achieve the typical 128-bit security level requirement in cryptographic protocols.  
Experiments are run on a desktop computer running the 20.04 LTS version of Ubuntu operating system with 64GB of memory, and 16 x 11th generation Intel i9@ 3.50 GHz cores.
Since the protocol operates entirely on one machine, there is no network-related delay. 
In addition, our implementation uses the caching mechanism described previously.
We compute the common neighbors heuristic for each pair of nodes in the graph.  
We also re-implemented the solution of~\cite{Demirag23} as described by its authors and present their performance next to ours in Table \ref{tab:performances}, with results obtained being consistent with those presented in their original paper.
Three scenarios are considered : (1) predictions on all the node pairs in the graph (all vs all), (2) predictions between a single random node and all the other ones (one vs all) and (3) predictions between two random nodes (one vs one). 
These results show a drastic improvement of one to several orders of magnitude in both computing time and communication. 

\begin{table}
    \centering 
    \begin{tabular}{lccccc}
        \hline
        Topology & Protocol & Time $G_1$ & Time $G_2$ & Com. $G_1$ & Com $G_2$ \\
        \hline
         all vs all & \cite{Demirag23} & 2h 21min 56s & 1h 14min 18s & 291.96MB & 316.07MB\\         & \textbf{\cryptograph{}}  & \textbf{5min 58s}  & \textbf{1s} & \textbf{0.70MB}  & \textbf{1.43MB}  \\
         \hline

          one vs all & \cite{Demirag23} & 28s & 28s & 20.44MB & 22.0MB\\
         & \textbf{\cryptograph{}} & \textbf{913ms} & \textbf{790ms} & \textbf{0.71MB} & \textbf{1.43MB} \\
        \hline
          one vs one & \cite{Demirag23} & 426ms & 432s & 0.31MB & 0.33MB\\
         & \textbf{\cryptograph{}} & \textbf{26ms} & \textbf{22ms} & \textbf{0.02MB} & \textbf{0.04MB}  \\
         \hline
    \end{tabular}
    \caption[Online running time and communication performances of \cryptograph, compared to \cite{Demirag23}]{Online running time and communication performances of \cryptograph, compared to \cite{Demirag23}}
    \label{tab:performances}
\end{table}

\subsection{Security Model}
\label{sec:security}

In this subsection, we introduce the security model of our protocol and prove its security against a semi-honest adversary\footnote{The term semi-honest adversary refers to a participant of the protocol that does not deviate maliciously from it but tries to infer new knowledge about the inputs of other parties from the information it gathers.} under well-known cryptographic assumptions that we recall hereafter. 
In the following, we denote the security parameter as $\lambda$ and $n$ as the number of nodes. 

\begin{definition}[Discrete Logarithm Assumption]
Let $\G$ be a cyclic group of generator $g$. 
The Discrete Logarithm Problem (DLP) is hard in $\G$ if, for every efficient algorithm $\A$, the following probability is a negligible function of $\lambda$ :
$$ \P[\A(g,g^a)=a]$$
\end{definition}

\begin{definition}[Decision Diffie-Hellman Assumption]
Let $\G$ be a cyclic group and $g$ be its generator. 
We assume that the bit-length of group size is $l$. 
The Decision Diffie-Hellman (DDH) problem is hard in $\G$ if, for every efficient algorithm $\A$, the following probability is a negligible function of $\lambda$:
$$|\P[x,y\gets\{0,1\}^l :\A(g,g^x,g^y,g^{xy})=1] - \P[x,y,z\gets\{0,1\}^l :A(g,g^x,g^y,g^z)=1]|$$
\end{definition}

\begin{definition}[One-More-Diffie-Hellman Assumption]
Let $\G$ be a cyclic group of order $q$ and $g$ be its generator. 
The One-More-DH problem~\cite{bellare2003one} is said to be ($\tau$, $t$)-hard if for every algorithm $\A$ that runs in time $t$ we have:
$$\P[\{(g_i,(g_i)^x)\}_{i=1,\dots,n+1} \gets A^{DH_x(.)}(g_1,\dots ,g_{m})] \leq \tau$$
in which $m \geq n$ and $\A^{DH_x(.)}$ is the algorithm $\A$ with access to a "$DH_x(.)$" oracle. 
We assume that $A$ can make at most $n$ queries to the $DH_x(.)$ oracle.
\end{definition}

\begin{theorem}
The security of the proposed protocol is ensured by the DLP problem and the One-More-DH problem if both parties reinitialise their keys $\alpha$ and $\beta$ after each instantiation of the protocol.
\end{theorem}

\begin{proof}
Let $G=(\V,\E)$ denote a graph shared between two parties $P_1$ and $P_2$. 
Let $x\in \V$ be a node in $G$ and lets denote $x_i$ the elements of $\Gamma_1(x)$ and $x_i'$ the elements of $\Gamma_2(x)$. 
During the protocol, $P_2$ obtains the elements from $P_1$ in the encrypted form $a_i=g^{x_i\alpha}$ that $P_2$ cannot decrypt without $\alpha$ because of the DLP problem. 
Afterwards, $P_2$ encrypts his own elements and sends them in the form $c_i= g^{x_i'\beta}$ before encrypting the elements of $P_1$ by sending them in the form $a_i'=g^{x_i\alpha\beta}$. 
From here, several scenarios are possible:
\begin{enumerate}
\item If there are no elements in the intersection, then the elements of $P_2$ are protected by the hardness of DLP. 
\item If there is only one element in the intersection, for instance $x_j=x'_l$ so $P_1$ can get $g^{x_j\beta} = g^{x'_l\beta}$ (by doing a modular exponentiation of $\alpha^{-1} \in \Z_q$). 
The hardness of DLP implies that $P_1$ cannot find an algorithm $\A$ that run in polynomial time to recover $x_j\beta$. Furthermore, even if we discard the DLP assumption, $x_j\beta$ is indistinguishable from a random element of $\Z_q$.
\item If there are several elements in the intersection, say $x_{j_0}=x'_{l_0},\dots, x_{j_m}=x'_{l_m}$ so $P_1$ can get $g^{x_{j_1}\beta},\dots, g^{x_{j_m}\beta}$ along with $g_1=g^{x_{j_1}},\dots, g_m=g^{x^{j_{m}}}$ that matches the One-More-DH assumption. 
Indeed, this scenario arises due to the following reasoning: if $P_1$ possesses a $DH_x(.)$ oracle facilitating the retrieval of $m$ pairs, namely $(g_1,g_1^\beta), \dots ,(g_m,g_m^\beta)$, it is infeasible for $P_1$ to devise an algorithm $\A$ that run in polynomial time and aims to successfully recover an additional pair $(g_t,g_t^{\beta})$, in which $0 \leq t \leq m$. 
The intuition behind this is to exploit $g_t$ in order to access the corresponding element within the intersection set. Thus, $P_2$'s privacy is ensured by the One-More-DH assumption.
\end{enumerate}
\end{proof}

We have demonstrated that the neighborhood's privacy of the two nodes is preserved during the execution of \cryptograph{} under cryptographic assumptions. 
One might wonder if this privacy guarantee could be compromised when we apply the protocol to all pairs of nodes. 
Indeed, an honest-but-curious adversary could attempt to infer information from the number of neighbors of nodes that have been already considered during the protocol and thereby try to reconstruct the common graph. 

\begin{theorem}
The proposed protocol applied to all nodes is secure against a semi-honest adversary and the worst-case complexity to recover the entire graph is $\mathcal{O}(2^n\sqrt{n})$, in which $n$ is the number of the nodes of the graph.
\end{theorem}

\begin{proof}
Let $G=(\V,\E)$ denote a graph shared between two parties $P_1$ and $P_2$.
We have seen that \cryptograph{} performs the PSI-CA described in Figure~\ref{fig:privatelinkprediction} on all pairs of nodes in the graphs of $P_1$ denoted as $G_1$ and $P_2$ denoted as $G_2$. 
Focusing on $P_1$ performing a PSI-CA between $x$ and all $x_i \in \mathcal{V} \setminus \{x\}$, it does $n-1$ PSI-CA and sequentially receives $n-1$ outcomes ranging from $0$ to $n-2$ (depending on the extent of shared nodes between $x$ and $x_i$ in $G_2$). 
In order to assess the extent of information that $P_1$ can deduce while executing the protocol, one possibility is to conduct a brute-force attack. 
This attack enables to derive an upper bound on the cost incurred by $P_1$ in reconstructing the merged graph $G$ or at  least in determining the neighboring nodes of $x \in \V$. 
Thus, $P_1$ has $|\Gamma(x) \cap \Gamma(x_i)|$ for each $ x_i \in \V / \{x\}$. 
Therefore in the worst case, $P_1$ has $$\binom{n-2}{\frac{n-2}{2}}^{n-1}$$ possible ways of reconstructing the neighborhood of $x$, which can be bounded by $\frac{2^{n-2}}{\sqrt{n-2}}$ according to Stirling's approximation. 
As a consequence the worst-case complexity of the brute-force attack that aims at identifying the neighborhood of a specific node is $\mathcal{O}(\frac{2^{n}}{\sqrt{n}})$. 
By extending this analysis to encompass all nodes, the resulting complexity is $\mathcal{O}(2^{n}\sqrt{n})$. 
\end{proof}

\section{Application to graph sanitization}
\label{sec:application-gnn}

Leveraging on our protocol, we can design a privacy-preserving defense mechanism against attempts to poison data in a GNN application. 
Hereafter, we provide significant experimental evidences of the effectiveness of this defense against state-of-the art graph attacks.

\subsection{Description}
Our approach acts as a preprocessing step by helping to privately clean the distributed graph before downstream learning tasks. 
The core idea of our approach is to identify suspicious or unlikely edges in the distributed graph based on the same approach as link prediction, which we refer to \textit{link removal}. 
These suspicious connections could be indicative of malicious intent as it has been shown in~\cite{Wu19}. 
To obtain a likelihood score for each of the links on the joint network, we run \cryptograph{} for all the possible pairs of nodes in the network by computing the similarity measures introduced in Section~\ref{sec:linkprediction}. 
A threshold $t_i$ is then set by each party $i$, such that all links between pairs of nodes that have a similarity below $t_i$ are considered malicious and discarded from $G_i$.

\subsection{Experimental evaluation}

To assess the benefit of our collaborative approach over individual defense strategies, we evaluate its effectiveness against various types of attacks. 
We consider targeted attacks that aim at changing the class of specific nodes in the graph as well as global attacks that try to decrease the global classification accuracy over all the nodes. 
More precisely, we measure the performance of our defense against the IG-FGSM \cite{Wu19}, Nettack \cite{Zugner18} and Dice \cite{zugner19} attacks. 
The FGSM attack, a targeted attack traditionally applied to continuous image data, has been adapted to the discrete graph context by Wu and collaborators with the use of integrated gradient, hence the name IG-FGSM. 
Zugner and colleagues have proposed Nettack, another targeted attack using gradients to identify high-impact links and maliciously inject them in the neighborhood of an attacked node. 
As of Dice, it is introduced in~\cite{zugner19} as a baseline global attack, which simply randomly creates links between nodes belonging to different classes while removing links between nodes of the same class. 

\begin{figure}[hbt!]
    \centering
    \includegraphics[width=\linewidth]{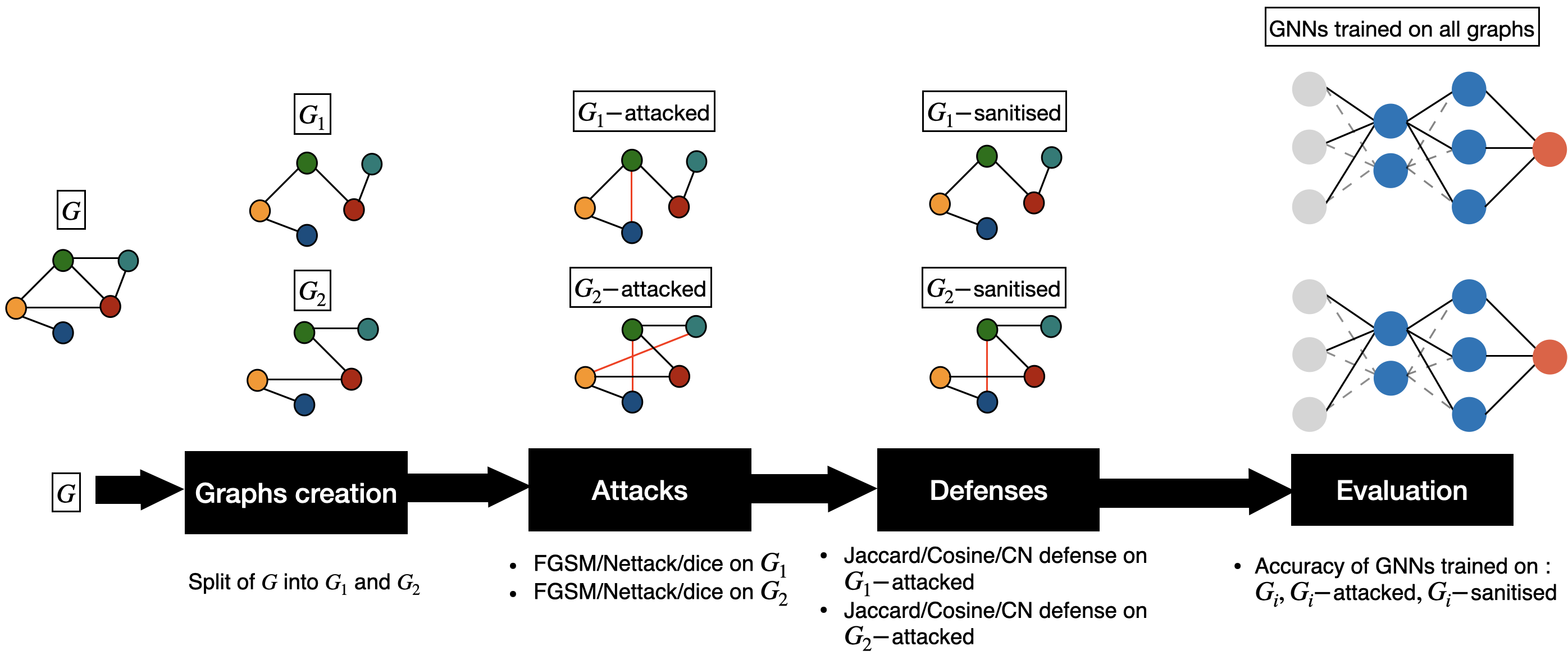}
    \caption{Experimental pipeline for the application of \cryptograph{ } to graph sanitization. 
    At the end of the pipeline, we train simultaneously two GNNs on the different versions of $G_1$ and $G_2$ given through the pipeline.}
    \label{fig:pipeline}
\end{figure}

Figure~\ref{fig:pipeline} represents the experimental pipeline used in this section. 
Our experiments start by creating two subgraphs from the Polblogs dataset, as described in Section \ref{protocol:performance}.
Afterwards, we inject malicious links into $G_1$ and $G_2$ using the previous attacks. 
Since Nettack and IG-FGSM are targeted attacks, we select 20 nodes that will undergo these attacks, while this is not needed for Dice. 
Finally, we apply the three different defense strategies on the poisoned graphs before evaluating the performance of GNNs trained on the sanitized graphs produced by these defenses. 

We have identified the following key parameters that might influence the outcome of the defense mechanisms and subsequently the accuracy of the GNNs trained on sanitized data:
\begin{itemize}
    \item \textbf{The thresholds of similarity $t_1, t_2$ for link removal}.
    This parameter directly impacts the number of links removed during the sanitization. 
    Indeed a high threshold might induce a high false positive rate whereas a low threshold could leave malicious links in the graphs (\emph{i.e.}, leading to false negatives). 
    Note that for party $i$, $t_i \in [0,1]$ for the Jaccard and cosine similarities, and $t_i \in \mathbb{N}$ for the common neighbors similarity.
    \item \textbf{The common proportion of links $ppt$ in the two graphs}. 
    Since the graphs can have overlapping knowledge about the global network, our assumption is that the less they share, the more a collaborative defense is effective.
    \item \textbf{The perturbation rates $r_1$, $r_2$ of the attacks}. 
    This factor influences how many malicious links are introduced by the attacks. 
    More precisely, a perturbation rate of $r_i$ on a certain node $x$ means that the attack is allowed to add at most $r_i \times d(x)$ links to node $x$, in which $d(x)$ is the degree of node $x$. 
    In contrast, a $r_i$-Dice attack on $G_i$ means that the attack injects exactly $r_i \times|\E_i|$ malicious links in the entire graph. 
    For this parameter, our hypothesis is that the more the graphs are attacked, the harder it becomes to recover from such perturbations. 
\end{itemize} 

To demonstrate the protection potential of \cryptograph{}, we have studied the performances of the node classification task on the Polblogs dataset in relation to different ranges of previous parameters. 
More precisely, we begin by exploring the impact of the similarity threshold, before evaluating the variation of the shared proportion of data and finishing with the experiments on the perturbation rate.

\subsubsection{Impact of similarity threshold for defense}
In this section, we study the accuracy of GNNs trained on graphs after local and distributed defenses with different values of $t$. 

\begin{figure}[hbt!]
    \centering
    \includegraphics[width=\linewidth]{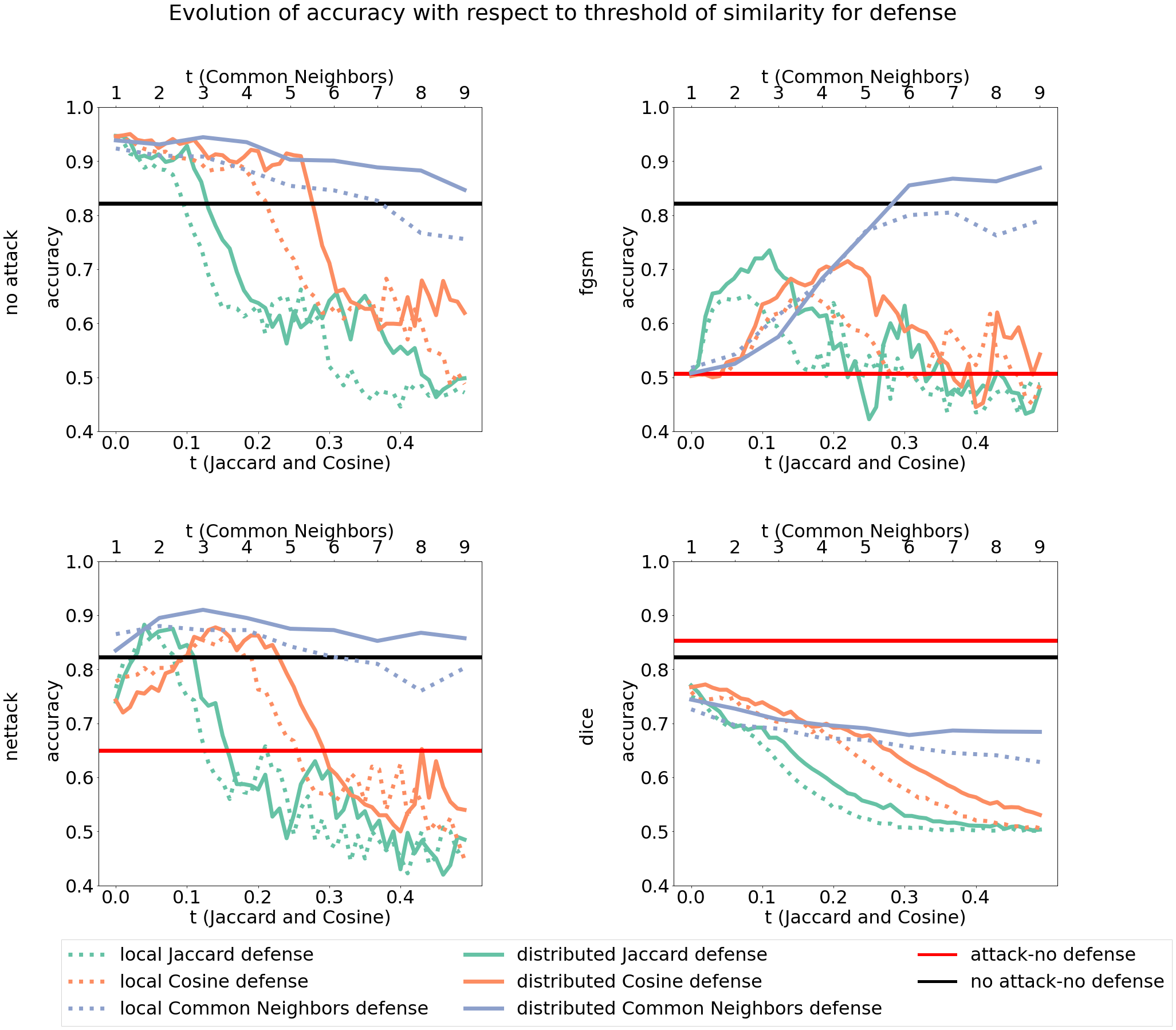}
    \caption{Impact of the similarity threshold $t$ for link removal on the accuracy of GNNs trained on sanitized graphs. The Jaccard, Cosine and Common Neighbors based defenses are presented in the context of no attack, the FGSM, Nettack as well as Dice attacks.} 
    \label{fig:impact_of_threshold}
\end{figure}

More precisely, in this series of experiments, the common proportion $ppt$ of links between $G_1$ and $G_2$ is set to 0.5, $r_1 = 0$ and $r_2 = 0.5$ (\emph{i.e.}, they both share 50\% of the data and each own 25\% of fresh and original data).
Figure~\ref{fig:impact_of_threshold} represents the evolution of the average accuracy of GNNs trained on $G_1$ and $G_2$ in three settings: when no attack and no defense have been performed, on the attacked graphs without defense and on the attacked and sanitized graphs with different defense mechanisms.

We summarize our findings on the impact of the similarity thresholds as follows:
\begin{itemize}
    \item As expected, different thresholds lead to different graph qualities, which in turn induces varying performances for the trained GNNs. The same remark can be made for the similarity metric used for defense.
    \item The distributed defense mechanisms tend to be better than their local counterparts for most of the thresholds, especially the optimal one.
    \item In some situations, the defense mechanisms even allow for a better performance than on the clean graphs. 
    We believe that this could be due to the fact that the defense removes the outliers from the data and that this helps for the GNNs tasks. This is an interesting finding that can motivate the usage of such defenses even when it is not clear if the graphs have been attacked.
    \item The Dice attack surprisingly improves the quality of the graphs. 
    To understand this, one should remember that this attack is really simplistic, and thus might actually add absent but likely connections in the graphs, making them more useful for learning.
\end{itemize}

\subsubsection{Impact of shared proportion of links}
Since the global network can be distributed in many ways, we study the impact of the distribution of links over $G_1$ and $G_2$.

\begin{figure}[hbt!]
    \centering
    \includegraphics[width=\linewidth]{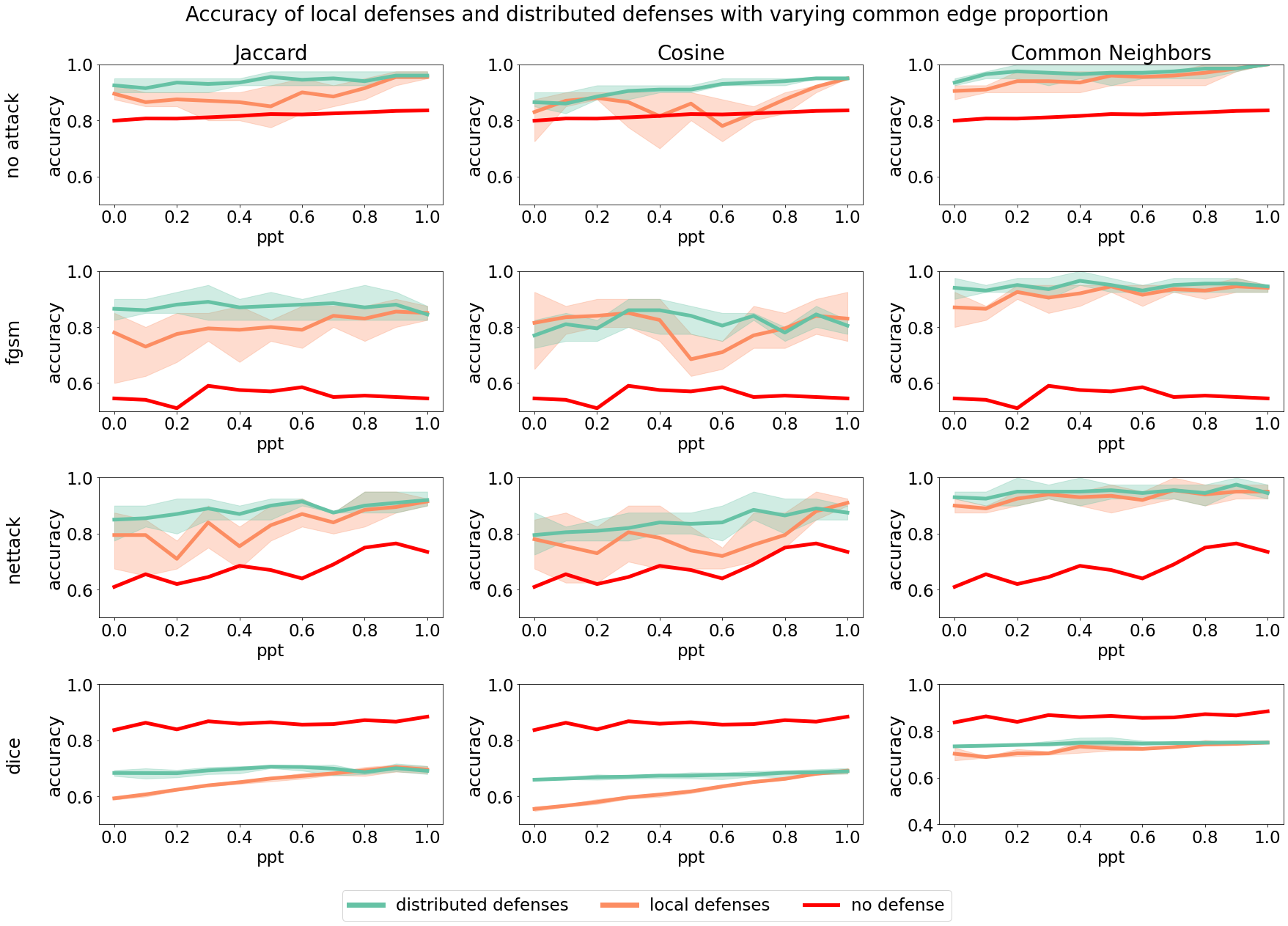}
    \caption{Impact of the shared proportion (ppt) of links between $G_1$ and $G_2$ on the accuracy of the defense based on the different similarity metrics (Jaccard, Cosine and Common Neighbors), depending on the type of attack used (no, FGSM, Nettack and Dice).}
    \label{fig:impact_of_threshold_2d}
\end{figure}

Namely, we vary the proportion $ppt$ of common links owned by $G_1$ and $G_2$ such that $ppt$ ranges from 0 ($\E_1 \cap \E_2 = \emptyset $) to 1 ($\E_1 = \E_2)$. $t_1$ and $t_2$ are set to be the thresholds providing the best accuracy for each metric, $r_1 = 0$ and $r_2 = 0.5$. 
As before, three settings are considered for the evaluation : no attack and no defense have been conducted, attacks have been performed without subsequent defense and finally attack and defense have been deployed.

\begin{figure}[hbt!]
    \centering
    \includegraphics[scale=0.30]{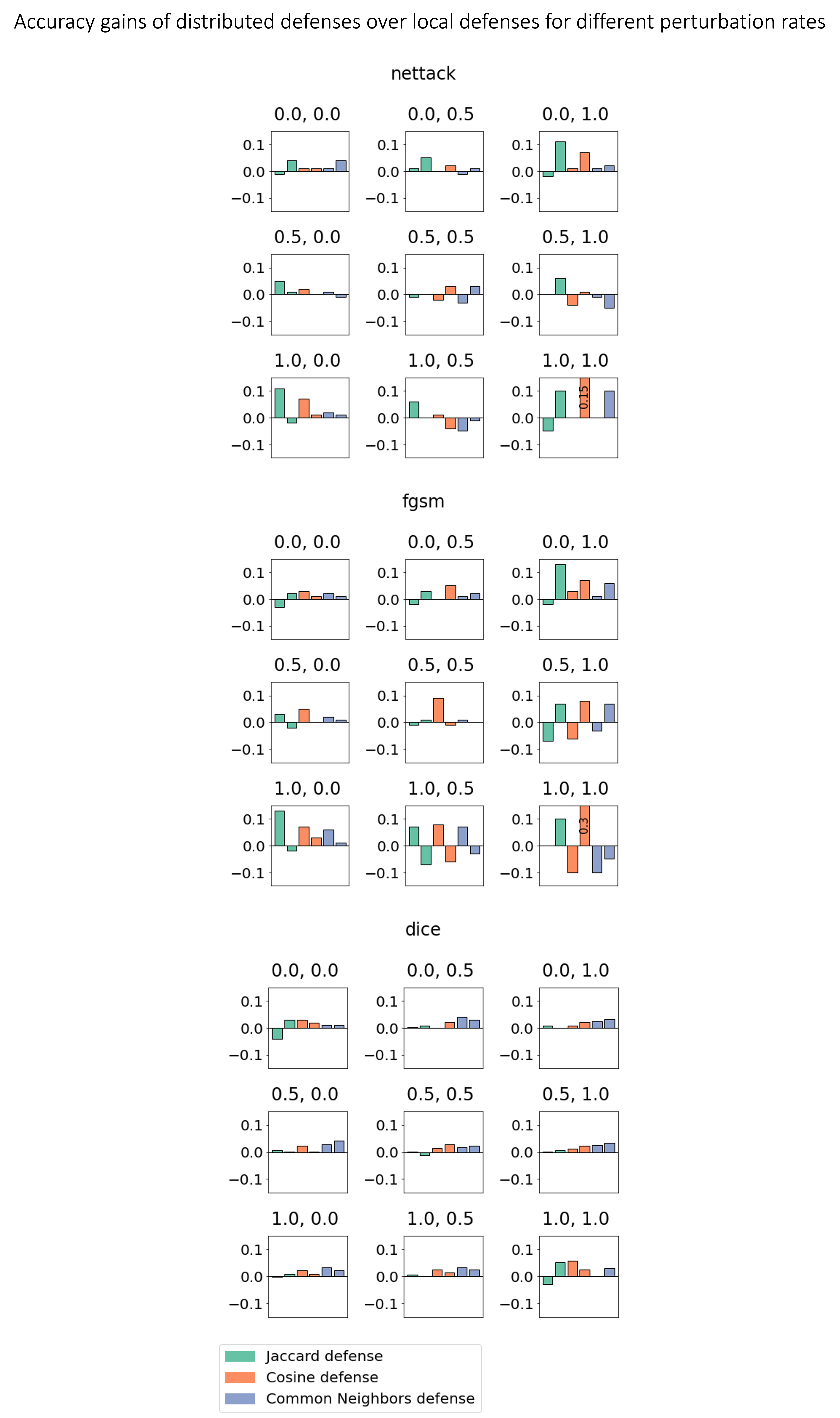}
    \caption{Impact of the perturbation rate of the attacks on the accuracy of a GNN trained on sanitized graph. 
    The amount of perturbation is denoted as $(r_1,r_2)$ on top of histograms. 
    The first and second bar of each metric (\emph{i.e}, each color) represents the accuracy gain for $G_1$ and $G_2$.}
    \label{fig:impact_of_perturbation_rate}
\end{figure}

From the results of this series of experiments, we derive the following observations:
\begin{itemize}
    \item The distributed defense metrics stay consistently better than the local ones over the whole range of proportions.
    \item The quality of the local defenses grows with the proportion and gets near (sometimes exceeds) the performance of distributed defenses, meaning that the more each graph already shares links with the other, the more they can get from a local defense. This directly matches our initial assumption.
    \item Overall, the distributed defenses are more constant than the local ones across multiple instances, which represents another advantage for them. 
    Indeed in practice, one is more likely to choose a defense with a high and stable quality overtime.
     \item Again, the defense mechanisms (especially the distributed ones) make the graphs even better than if they were not attacked in the case of the IG-FGSM and Nettack attacks
     \item As before, the Dice attack slightly improves the accuracy obtained on an attacked graph.
\end{itemize}

To demonstrate the performance of \cryptograph{{}} against different attack rates, we also evaluate the accuracy of GNNs trained on graphs sanitized after various forces of attack. 

\subsubsection{Impact of perturbation rates}

To realize this, we study the combinations of three perturbation rates $\{0, 0.5, 1\}$ for  $r_1$ and $r_2$ with $ppt = 0.5$. 
Here, we show the accuracy gain of the distributed metrics over their local equivalents, with a positive margin meaning that the distributed metric is better that the local one, whereas a negative one favors the local metric.
The results presented in Figure~\ref{fig:impact_of_perturbation_rate} lead us to the following conclusions:
\begin{itemize}
    \item Overall, each graph can find a positive margin for each of the scenarios. 
    This is especially important as the two graphs are not necessarily sanitized using the same defense metric, which leaves room for each graph owner to choose the metric that best suits them.
    \item Often, the most attacked graph is the one with the lowest accuracy gain, which validates our assumption that high perturbation rates are more difficult to overcome.
    \item A perturbation rate of 1 is extreme, but in many of the reasonable scenarios, it appears that it is not too costly for the least attacked graph to cooperate and that it is always beneficial for the most attacked one to do so.
\end{itemize}

\section{Conclusion}
\label{sec:conclusion}

In this article, we have proposed \cryptograph{}, a protocol for privacy-preserving distributed link prediction. \cryptograph{} is more efficient than the other state-of-the-art methods both in terms of computation and communication, by one to several orders of magnitude, while reaching exactly the same utility. 
Additionally, our protocol is able to compute different similarity metrics, allowing for data owners to choose the best one according to their specific needs.
We also demonstrate that \cryptograph{} is secure against external eavesdroppers and against honest-but-curious participants. 
Based on \cryptograph{}, we build a distributed defense mechanism against data poisoning in graph neural networks application scenarios. 
Our experiments show that this mechanism is effective to mitigate those attacks and can even be beneficial in the absence of attack.  
We also show that the more disjoint the data of the participants is, the more beneficial it is for them to cooperate via our distributed defense mechanism. 
Those benefits vary according to the power of the data poisoning attack. 
In reasonable attack scenarios, cooperation is a good strategy while it is a little more complex in extreme ones.

As a defense against the graph reconstruction attack presented above (which we have shown to be quite difficult to carry in the worst case in our security analysis), we consider as future works the application of methods like the occasional injection of dummy links in the graphs or an adaptation of differential privacy in our context. 
We would also like to propose a method for a private and efficient choice of the defense threshold, which we assume known by each of the party in our current solution. 
Finally in another direction, we would like to better understand the security of our method by exploring more in-depth attack strategies and also combining several similarity metrics to better counter these attacks.

\bibliography{references}

\begin{thebibliography}{10}

\bibitem{sec2}
{\em SEC 2: Recommended Elliptic Curve Domain Parameters}, 2010.

\bibitem{Adamic03}
Lada~A Adamic and Eytan Adar.
\newblock Friends and neighbors on the web.
\newblock {\em Social Networks}, 25(3):211--230, 2003.
\newblock URL: \url{https://www.sciencedirect.com/science/article/pii/S0378873303000091}, \href {https://doi.org/10.1016/S0378-8733(03)00009-1} {\path{doi:10.1016/S0378-8733(03)00009-1}}.

\bibitem{antweiler2022uncovering}
Dario Antweiler, David Sessler, Maxim Rossknecht, Benjamin Abb, Sebastian Ginzel, and J{\"o}rn Kohlhammer.
\newblock Uncovering chains of infections through spatio-temporal and visual analysis of covid-19 contact traces.
\newblock {\em Computers \& Graphics}, 106:1--8, 2022.

\bibitem{Aziz17}
Md. Momin~Al Aziz, Dima Alhadidi, and Noman Mohammed.
\newblock Secure approximation of edit distance on genomic data.
\newblock {\em BMC Medical Genomics}, 10, 07 2017.
\newblock \href {https://doi.org/10.1186/s12920-017-0279-9} {\path{doi:10.1186/s12920-017-0279-9}}.

\bibitem{bellare2003one}
Bellare, Namprempre, Pointcheval, and Semanko.
\newblock The one-more-rsa-inversion problems and the security of chaum's blind signature scheme.
\newblock {\em Journal of Cryptology}, 16:185--215, 2003.

\bibitem{decristofaro12}
Emiliano De~Cristofaro, Paolo Gasti, and Gene Tsudik.
\newblock Fast and private computation of cardinality of set intersection and union.
\newblock In {\em International Conference on Cryptology and Network Security}, pages 218--231. Springer, 2012.

\bibitem{Debnath17}
Sumit~Kumar Debnath and Ratna Dutta.
\newblock Provably secure fair mutual private set intersection cardinality utilizing bloom filter.
\newblock In Kefei Chen, Dongdai Lin, and Moti Yung, editors, {\em Information Security and Cryptology}, pages 505--525, Cham, 2017. Springer International Publishing.

\bibitem{Debnath20}
Sumit~Kumar Debnath, Pantelimon Stănică, Tanmay Choudhury, and Nibedita Kundu.
\newblock Post-quantum protocol for computing set intersection cardinality with linear complexity.
\newblock {\em IET Information Security}, 14(6):661--669, 2020.
\newblock URL: \url{https://ietresearch.onlinelibrary.wiley.com/doi/abs/10.1049/iet-ifs.2019.0315}, \href {https://arxiv.org/abs/https://ietresearch.onlinelibrary.wiley.com/doi/pdf/10.1049/iet-ifs.2019.0315} {\path{arXiv:https://ietresearch.onlinelibrary.wiley.com/doi/pdf/10.1049/iet-ifs.2019.0315}}, \href {https://doi.org/10.1049/iet-ifs.2019.0315} {\path{doi:10.1049/iet-ifs.2019.0315}}.

\bibitem{Demirag23}
Didem Demirag, Mina Namazi, Erman Ayday, and Jeremy Clark.
\newblock Privacy-preserving link prediction.
\newblock In Joaquin Garcia-Alfaro, Guillermo Navarro-Arribas, and Nicola Dragoni, editors, {\em Data Privacy Management, Cryptocurrencies and Blockchain Technology}, pages 35--50, Cham, 2023. Springer International Publishing.

\bibitem{Diffie76}
Whitfield Diffie and Martin Hellman.
\newblock New directions in cryptography.
\newblock {\em IEEE Transactions on Information Theory}, 22(6):644--654, 1976.
\newblock \href {https://doi.org/10.1109/TIT.1976.1055638} {\path{doi:10.1109/TIT.1976.1055638}}.

\bibitem{Dong17}
Changyu Dong and Grigorios Loukides.
\newblock Approximating private set union/intersection cardinality with logarithmic complexity.
\newblock {\em IEEE Transactions on Information Forensics and Security}, 12(11):2792--2806, 2017.
\newblock \href {https://doi.org/10.1109/TIFS.2017.2721360} {\path{doi:10.1109/TIFS.2017.2721360}}.

\bibitem{Freedman04}
Michael~J. Freedman, Kobbi Nissim, and Benny Pinkas.
\newblock Efficient private matching and set intersection.
\newblock In Christian Cachin and Jan~L. Camenisch, editors, {\em Advances in Cryptology - EUROCRYPT 2004}, pages 1--19, Berlin, Heidelberg, 2004. Springer Berlin Heidelberg.

\bibitem{Grover16}
Aditya Grover and Jure Leskovec.
\newblock Node2vec: Scalable feature learning for networks.
\newblock In {\em Proceedings of the 22nd ACM SIGKDD International Conference on Knowledge Discovery and Data Mining}, KDD '16, page 855–864, New York, NY, USA, 2016. Association for Computing Machinery.
\newblock \href {https://doi.org/10.1145/2939672.2939754} {\path{doi:10.1145/2939672.2939754}}.

\bibitem{Ion21}
Mihaela Ion, Ben Kreuter, Ahmet~Erhan Nergiz, Sarvar Patel, Shobhit Saxena, Karn Seth, Mariana Raykova, David Shanahan, and Moti Yung.
\newblock On deploying secure computing: Private intersection-sum-with-cardinality.
\newblock In {\em 2020 IEEE European Symposium on Security and Privacy (EuroS\&P)}, pages 370--389, 2020.
\newblock \href {https://doi.org/10.1109/EuroSP48549.2020.00031} {\path{doi:10.1109/EuroSP48549.2020.00031}}.

\bibitem{Jaccard01}
Paul Jaccard.
\newblock Etude de la distribution florale dans une portion des alpes et du jura.
\newblock {\em Bulletin de la Societe Vaudoise des Sciences Naturelles}, 1901.

\bibitem{Kashima06}
Hisashi Kashima and Naoki Abe.
\newblock A parameterized probabilistic model of network evolution for supervised link prediction.
\newblock In {\em Sixth International Conference on Data Mining (ICDM'06)}, pages 340--349, 2006.
\newblock \href {https://doi.org/10.1109/ICDM.2006.8} {\path{doi:10.1109/ICDM.2006.8}}.

\bibitem{Kipf17}
Thomas~N. Kipf and Max Welling.
\newblock Semi-supervised classification with graph convolutional networks.
\newblock In {\em 5th International Conference on Learning Representations, {ICLR} 2017, Toulon, France, April 24-26, 2017, Conference Track Proceedings}. OpenReview.net, 2017.
\newblock URL: \url{https://openreview.net/forum?id=SJU4ayYgl}.

\bibitem{Liben03}
David Liben-Nowell and Jon Kleinberg.
\newblock The link prediction problem for social networks.
\newblock In {\em Proceedings of the Twelfth International Conference on Information and Knowledge Management}, CIKM '03, page 556–559, New York, NY, USA, 2003. Association for Computing Machinery.
\newblock \href {https://doi.org/10.1145/956863.956972} {\path{doi:10.1145/956863.956972}}.

\bibitem{Liu18}
Dongxiao Liu, Jianbing Ni, Hongwei Li, Xiaodong Lin, and Xuemin Shen.
\newblock Efficient and privacy-preserving ad conversion for v2x-assisted proximity marketing.
\newblock In {\em 2018 IEEE 15th International Conference on Mobile Ad Hoc and Sensor Systems (MASS)}, pages 10--18, 2018.
\newblock \href {https://doi.org/10.1109/MASS.2018.00014} {\path{doi:10.1109/MASS.2018.00014}}.

\bibitem{Siyi20}
Siyi Lv, Jinhui Ye, Sijie Yin, Xiaochun Cheng, Chen Feng, Xiaoyan Liu, Rui Li, Zhaohui Li, Zheli Liu, and Li~Zhou.
\newblock Unbalanced private set intersection cardinality protocol with low communication cost.
\newblock {\em Future Generation Computer Systems}, 102:1054--1061, 2020.
\newblock URL: \url{https://www.sciencedirect.com/science/article/pii/S0167739X19316413}, \href {https://doi.org/10.1016/j.future.2019.09.022} {\path{doi:10.1016/j.future.2019.09.022}}.

\bibitem{Mikolov13}
Tom{\'{a}}s Mikolov, Kai Chen, Greg Corrado, and Jeffrey Dean.
\newblock Efficient estimation of word representations in vector space.
\newblock In Yoshua Bengio and Yann LeCun, editors, {\em 1st International Conference on Learning Representations, {ICLR} 2013, Scottsdale, Arizona, USA, May 2-4, 2013, Workshop Track Proceedings}, 2013.
\newblock URL: \url{http://arxiv.org/abs/1301.3781}.

\bibitem{miyaji2017privacy}
Atsuko Miyaji, Kazuhisa Nakasho, and Shohei Nishida.
\newblock Privacy-preserving integration of medical data: a practical multiparty private set intersection.
\newblock {\em Journal of medical systems}, 41:1--10, 2017.

\bibitem{MORALES2023}
Daniel Morales, Isaac Agudo, and Javier Lopez.
\newblock Private set intersection: A systematic literature review.
\newblock {\em Computer Science Review}, 49:100567, 2023.
\newblock URL: \url{https://www.sciencedirect.com/science/article/pii/S1574013723000345}, \href {https://doi.org/10.1016/j.cosrev.2023.100567} {\path{doi:10.1016/j.cosrev.2023.100567}}.

\bibitem{Nagaraja10}
Shishir Nagaraja, Prateek Mittal, Chi-Yao Hong, Matthew Caesar, and Nikita Borisov.
\newblock {BotGrep}: Finding {P2P} bots with structured graph analysis.
\newblock In {\em 19th USENIX Security Symposium (USENIX Security 10)}, Washington, DC, August 2010. USENIX Association.
\newblock URL: \url{https://www.usenix.org/conference/usenixsecurity10/botgrep-finding-p2p-bots-structured-graph-analysis}.

\bibitem{Newman01}
M.~E.~J. Newman.
\newblock Clustering and preferential attachment in growing networks.
\newblock {\em Phys. Rev. E}, 64:025102, Jul 2001.
\newblock \href {https://doi.org/10.1103/PhysRevE.64.025102} {\path{doi:10.1103/PhysRevE.64.025102}}.

\bibitem{Nomura20}
Kenta Nomura, Yoshiaki Shiraishi, Masami Mohri, and Masakatu Morii.
\newblock Secure association rule mining on vertically partitioned data using private-set intersection.
\newblock {\em IEEE Access}, 8:144458--144467, 2020.
\newblock \href {https://doi.org/10.1109/ACCESS.2020.3014330} {\path{doi:10.1109/ACCESS.2020.3014330}}.

\bibitem{pavlopoulos2011using}
Georgios~A Pavlopoulos, Maria Secrier, Charalampos~N Moschopoulos, Theodoros~G Soldatos, Sophia Kossida, Jan Aerts, Reinhard Schneider, and Pantelis~G Bagos.
\newblock Using graph theory to analyze biological networks.
\newblock {\em BioData mining}, 4:1--27, 2011.

\bibitem{Perozzi14}
Bryan Perozzi, Rami Al-Rfou, and Steven Skiena.
\newblock Deepwalk: Online learning of social representations.
\newblock In {\em Proceedings of the 20th ACM SIGKDD International Conference on Knowledge Discovery and Data Mining}, KDD '14, page 701–710, New York, NY, USA, 2014. Association for Computing Machinery.
\newblock \href {https://doi.org/10.1145/2623330.2623732} {\path{doi:10.1145/2623330.2623732}}.

\bibitem{ResendeA18}
Amanda Cristina~Davi Resende and Diego~F. Aranha.
\newblock Faster unbalanced private set intersection.
\newblock In {\em Financial Cryptography}, volume 10957, pages 203--221. Springer, 2018.

\bibitem{nr}
Ryan~A. Rossi and Nesreen~K. Ahmed.
\newblock The network data repository with interactive graph analytics and visualization.
\newblock In {\em AAAI}, 2015.
\newblock URL: \url{https://networkrepository.com}.

\bibitem{Ruan19}
Ou~Ruan, Zihao Wang, Jing Mi, and Mingwu Zhang.
\newblock New approach to set representation and practical private set-intersection protocols.
\newblock {\em IEEE Access}, 7:64897--64906, 2019.
\newblock \href {https://doi.org/10.1109/ACCESS.2019.2917057} {\path{doi:10.1109/ACCESS.2019.2917057}}.

\bibitem{Scarselli09}
Franco Scarselli, Marco Gori, Ah~Chung Tsoi, Markus Hagenbuchner, and Gabriele Monfardini.
\newblock The graph neural network model.
\newblock {\em IEEE Transactions on Neural Networks}, 20(1):61--80, 2009.
\newblock \href {https://doi.org/10.1109/TNN.2008.2005605} {\path{doi:10.1109/TNN.2008.2005605}}.

\bibitem{Shen18}
Liyan Shen, Xiaojun Chen, Dakui Wang, Binxing Fang, and Ye~Dong.
\newblock Efficient and private set intersection of human genomes.
\newblock In {\em 2018 IEEE International Conference on Bioinformatics and Biomedicine (BIBM)}, pages 761--764, 2018.
\newblock \href {https://doi.org/10.1109/BIBM.2018.8621291} {\path{doi:10.1109/BIBM.2018.8621291}}.

\bibitem{Singhal01}
Amit Singhal.
\newblock Modern information retrieval: A brief overview.
\newblock {\em Bulletin of the IEEE Computer Society Technical Committee on Data Engineering}, 2001.

\bibitem{Wilson09}
Christo Wilson, Bryce Boe, Alessandra Sala, Krishna~P.N. Puttaswamy, and Ben~Y. Zhao.
\newblock User interactions in social networks and their implications.
\newblock In {\em Proceedings of the 4th ACM European Conference on Computer Systems}, EuroSys '09, page 205–218, New York, NY, USA, 2009. Association for Computing Machinery.
\newblock \href {https://doi.org/10.1145/1519065.1519089} {\path{doi:10.1145/1519065.1519089}}.

\bibitem{Wu19}
Huijun Wu, Chen Wang, Yuriy Tyshetskiy, Andrew Docherty, Kai Lu, and Liming Zhu.
\newblock Adversarial examples for graph data: Deep insights into attack and defense.
\newblock In {\em Proceedings of the Twenty-Eighth International Joint Conference on Artificial Intelligence, {IJCAI-19}}, pages 4816--4823. International Joint Conferences on Artificial Intelligence Organization, 7 2019.
\newblock \href {https://doi.org/10.24963/ijcai.2019/669} {\path{doi:10.24963/ijcai.2019/669}}.

\bibitem{Xu23}
Xiaojun Xu, Hanzhang Wang, Alok Lal, Carl~A. Gunter, and Bo~Li.
\newblock Edog: Adversarial edge detection for graph neural networks.
\newblock In {\em 2023 IEEE Conference on Secure and Trustworthy Machine Learning (SaTML)}, pages 291--305, 2023.
\newblock \href {https://doi.org/10.1109/SaTML54575.2023.00027} {\path{doi:10.1109/SaTML54575.2023.00027}}.

\bibitem{Zhang23}
Hai-Feng Zhang, Xiao-Jing Ma, Jing Wang, Xingyi Zhang, Donghui Pan, and Kai Zhong.
\newblock Privacy-preserving link prediction in multiple private networks.
\newblock {\em IEEE Transactions on Computational Social Systems}, 10(2):538--550, 2023.
\newblock \href {https://doi.org/10.1109/TCSS.2022.3168010} {\path{doi:10.1109/TCSS.2022.3168010}}.

\bibitem{Zhang18}
Muhan Zhang and Yixin Chen.
\newblock Link prediction based on graph neural networks.
\newblock In {\em Proceedings of the 32nd International Conference on Neural Information Processing Systems}, NIPS'18, page 5171–5181, Red Hook, NY, USA, 2018. Curran Associates Inc.

\bibitem{Zheng15}
Yao Zheng, Bing Wang, Wenjing Lou, and Y.~Thomas Hou.
\newblock Privacy-preserving link prediction in decentralized online social networks.
\newblock In G{\"u}nther Pernul, Peter Y~A~Ryan, and Edgar Weippl, editors, {\em Computer Security -- ESORICS 2015}, pages 61--80, Cham, 2015. Springer International Publishing.

\bibitem{Zugner18}
Daniel Z\"{u}gner, Amir Akbarnejad, and Stephan G\"{u}nnemann.
\newblock Adversarial attacks on neural networks for graph data.
\newblock In {\em Proceedings of the 24th ACM SIGKDD International Conference on Knowledge Discovery \& Data Mining}, KDD '18, page 2847–2856, New York, NY, USA, 2018. Association for Computing Machinery.
\newblock \href {https://doi.org/10.1145/3219819.3220078} {\path{doi:10.1145/3219819.3220078}}.

\bibitem{zugner19}
Daniel Z{\"{u}}gner and Stephan G{\"{u}}nnemann.
\newblock Adversarial attacks on graph neural networks via meta learning.
\newblock In {\em 7th International Conference on Learning Representations, {ICLR} 2019, New Orleans, LA, USA, May 6-9, 2019}. OpenReview.net, 2019.
\newblock URL: \url{https://openreview.net/forum?id=Bylnx209YX}.

\end{thebibliography}

\end{document}